\documentclass[a4paper]{article}

\usepackage{amsmath}
\usepackage{amssymb}
\usepackage{color}
\usepackage{graphicx}
\usepackage{tikz}
\usepackage{xspace}
\usepackage{multirow}
\usepackage{amsmath, amsthm}
\usetikzlibrary{shapes,snakes}

\usepackage[color=green!40]{todonotes}
\usepackage[pdfdisplaydoctitle,menucolor=orange!40!black,filecolor=magenta!40!black,urlcolor=blue!40!black,linkcolor=red!40!black,citecolor=green!40!black,colorlinks]{hyperref}

\usepackage[noend]{algpseudocode}
\usepackage{algorithm}
\usepackage{aliascnt}

\newcommand{\mynewtheorem}[2]{
  \newaliascnt{#1}{dummy}
  \newtheorem{#1}[#1]{#2}
  \aliascntresetthe{#1}
  \expandafter\def\csname #1autorefname\endcsname{#2}
}

\theoremstyle{plain}
  \mynewtheorem{theorem}{Theorem}
  \mynewtheorem{proposition}{Proposition}
  \mynewtheorem{corollary}{Corollary}
  \mynewtheorem{lemma}{Lemma}
  \mynewtheorem{consequence}{Consequence}

\newtheorem{definition}{Definition}

\makeatletter
\def\NAT@spacechar{~}
\makeatother

\newcommand{\influence}[1][k,\ell]{{\sc{\ensuremath{(#1)}-Influence}}\xspace}
\newcommand{\conninfluence}{{\sc{Connected $(k,\ell)$-Influence}}\xspace}
\newcommand{\maxopinfluence}{{\sc{Max Open $k$-Influence}}\xspace}
\newcommand{\maxclinfluence}{{\sc{Max Closed $k$-Influence}}\xspace}
\newcommand{\maxclkplusuninfluence}{{\sc{Max Closed $(k+1)$-Influence}}\xspace}

\newcommand{\ids}{{\sc{Independent Dominating Set}}\xspace}
\newcommand{\mcs}{{\sc{Monotone Circuit Satisfiability}}\xspace}

\newcommand{\ds}{{\sc{Dominating Set}}\xspace}

\newcommand{\clique}{{\sc{Clique}}\xspace}
\newcommand{\densestksubgraph}{{\sc{Densest $k$-Subgraph}}\xspace}
\newcommand{\minvc}{{\sc{Min Vertex Cover}}\xspace}
\newcommand{\maxis}{{\sc{Max Independent Set}}\xspace}
\newcommand{\TSS}{\textsc{Target Set Selection}\xspace}

\newlength{\atextwidth}
\newcommand{\N}{\mathbb{N}}
\DeclareMathOperator{\thr}{thr}

\def \eps {\varepsilon}
\newcommand{\p}{\mathsf{P}}
\newcommand{\np}{\mathsf{NP}}
\newcommand{\zpp}{\mathsf{ZPP}}

\newcommand{\fpt}{\mathsf{FPT}}
\newcommand{\wone}{\mathsf{W[1]}}
\newcommand{\wtwo}{\mathsf{W[2]}}

\newcommand{\wx}[1]{\mathsf{W[#1]}}

\DeclareMathOperator{\Smax}{max}

{\rm\bfseries}{\itshape}

\newcommand{\mytabref}[1]{[\hyperref[#1]{Th.~\ref*{#1}}]}

\newcommand{\etal}{\textit{et al.}~}
\newcommand{\ie}{\textit{i.e.}~}

\tikzstyle{vertex}=[circle, fill=white, draw, inner sep=2pt, minimum width=4pt]
\pgfdeclarelayer{bg}
\pgfsetlayers{bg,main}

\newcommand{\problemdec}[3]{
  \vspace{1mm}
\noindent\fbox{
  \begin{minipage}{\atextwidth}
  \begin{tabular*}{\textwidth}{@{\extracolsep{\fill}}lr} #1 \\ \end{tabular*}
  {\bf{Input:}} #2  \\
  {\bf{Question:}} #3
  \end{minipage}
  }
  \vspace{1mm}
}

\newcommand{\problemopt}[3]{
  \vspace{1mm}
\noindent\fbox{
  \begin{minipage}{\atextwidth}
  #1 \\
  {\bf{Input:}} #2  \\
  {\bf{Output:}} #3
  \end{minipage}
  }
  \vspace{1mm}
}

\usepackage{etoolbox}

\newcommand{\appendixproof}[2]{
}

\newcommand{\qedfill}[0]{}

\begin{document}

\title{Parameterized Approximability of Maximizing the Spread of Influence in Networks\footnote{An extended abstract appeared in the  Proceedings of the 9th Annual International Computing and Combinatorics Conference (COCOON'13), volume 7936 of \emph{LNCS}, pages 543-554. Springer, 2013.}}


\author{Cristina Bazgan\footnote{LAMSADE - CNRS UMR 7243, PSL, Universit\'e Paris-Dauphine (France). \texttt{\{bazgan,chopin,florian.sikora\}@lamsade.dauphine.fr}}~~\footnote{\noindent Institut Universitaire de France}
\and Morgan Chopin$^{\dagger}$
\and Andr\'e Nichterlein\footnote{\noindent Institut f\"ur Softwaretechnik und Theoretische Informatik, TU Berlin, Germany. \texttt{andre.nichterlein@tu-berlin.de}}
\and Florian Sikora$^{\dagger}$
}

\date{}

\maketitle






\begin{abstract}
In this paper, we consider the problem of maximizing the spread of influence through a social network.
Given a graph with a threshold value~$\thr(v)$ attached to each vertex~$v$, the spread of influence is modeled as follows: A vertex~$v$ becomes ``active'' (influenced) if at least $\thr(v)$ of its neighbors are active.
In the corresponding optimization problem the objective is then to find a fixed number of vertices to activate such that the number of activated vertices at the end of the propagation process is maximum.
We show that this problem is strongly inapproximable in fpt-time with respect to (w.r.t.) parameter $k$ even for very restrictive thresholds. 
In the case that the threshold of each vertex equals its degree, we prove that the problem is inapproximable in polynomial time and it becomes $r(n)$-approximable in fpt-time w.r.t. parameter $k$ for any strictly increasing function $r$. 
 Moreover, we show that the decision version is $\wone$-hard w.r.t. parameter $k$ but becomes fixed-parameter tractable on bounded degree graphs.
\end{abstract}








\section{Introduction}
Optimization problems that involve a diffusion process in a graph are well studied \cite{peleg02,kempe2003,chen2009,aazami2009,dreyer2009,chang2009,benzwi2011,reddy2011}.
Such problems share the common property that, according to a specified \emph{propagation rule}, a chosen subset of vertices activates all or a fixed fraction of the vertices, where initially all but the chosen vertices are inactive.
Such optimization problems model the spread of influence or information in social networks via word-of-mouth recommendations, of diseases in populations, or of faults in distributed computing~\cite{peleg02,kempe2003,dreyer2009}.
One representative problem that appears in this context is the \textit{influence maximization} problem introduced by Kempe~\etal \cite{kempe2003}.
Given a directed graph, the task is to choose a fixed number of vertices such that the number of activated vertices at the end of the propagation process is maximized.
The authors show that the problem is polynomial-time ($\frac{e}{e-1} + \eps$)-approximable for any $\eps > 0$ under some stochastic propagation models, but $\mathsf{NP}$-hard to approximate within a ratio of $n^{1-\eps}$ for any $\eps > 0$ for general propagation rules.

In this paper, we use the following deterministic propagation model.
We are given an undirected graph, a threshold value~$\thr(v)$ associated to each vertex $v$, and the following propagation rule: a vertex~$v$ becomes active if at least~$\thr(v)$ many neighbors of~$v$ are active.
The propagation process proceeds in several rounds and stops when no further vertex becomes active.
Given this model, finding and activating a minimum-size vertex subset such that all or a fixed fraction of the vertices become active is known as the \textit{minimum target set selection} (MinTSS) problem introduced by Chen~\cite{chen2009}.
It has been shown $\np$-hard even for bipartite graphs of bounded degree when all thresholds are at most two~\cite{chen2009}.
Moreover, the problem was surprisingly shown to be hard to approximate within a ratio $O(2^{\log^{1-\eps}n})$ for any $\eps > 0$, even for constant degree graphs with thresholds at most two and for general graphs when the threshold of each vertex is half its degree (called \textit{majority} thresholds) \cite{chen2009}.
If the threshold of each vertex equals its degree (\textit{unanimity} thresholds), then the problem is polynomial-time equivalent to the vertex cover problem~\cite{chen2009} and, thus, admits a $2$-approximation and is hard to approximate with a ratio better than $1.36$~\cite{DS02}.
Concerning the parameterized complexity, the problem is shown to be~$\wtwo$-hard with respect to (w.r.t.) the solution size, even on bipartite graphs of diameter four with majority thresholds or thresholds at most two~\cite{NNUW12}.
Furthermore, it is~$\wone$-hard w.r.t. each of the parameters ``treewidth'', ``cluster vertex deletion number'', and ``pathwidth'' \cite{benzwi2011,CNNW12}.
On the positive side, the problem becomes fixed-parameter tractable w.r.t. each of the single parameters ``vertex cover number'', ``feedback edge set size'', and ``bandwidth'' \cite{NNUW12,CNNW12}.
If the input graph is complete, or has a bounded treewidth and bounded thresholds then the problem is polynomial-time solvable~\cite{NNUW12,benzwi2011}.

Here, \sloppy we study the maximization problem of MinTSS, called \textit{maximum $k$-influence} (Max$k$Inf) where the objective is to find $k$ vertices to activate such that the total number of activated vertices at the end of the propagation process is maximized.
Since both optimization problems have the same decision version, the parameterized as well as $\np$-hardness results directly transfer from MinTSS to Max$k$Inf.
We show that also Max$k$Inf is hard to approximate and, confronted with the computational hardness, we study the parameterized approximability of Max$k$Inf.

\paragraph{Our results}
Concerning the approximability of the problem, there are two possibilities of measuring the value of a solution: counting the vertices activated by the propagation process including or excluding the initially chosen vertices (denoted by \maxclinfluence and \maxopinfluence, respectively).
Observe that whether or not counting the chosen vertices might change the approximation factor. In this paper, we consider both cases and our approximability results are summarized in \autoref{tab:results}.
%

\begin{table}
\centering
\resizebox{\textwidth}{!}{
\begin{tabular}{|c|c|c|c|c|c|}
   \hline
    &   & \multicolumn{2}{c|}{\maxopinfluence} & \multicolumn{2}{c|}{\maxclinfluence} \\
   \hline
    Thresholds & Bounds & poly-time & fpt-time & poly-time & fpt-time \\
   \hline
   \multirow{2}{*}{General} & Upper & $n$ & $n$ & $n$ & $n$ \\
           & Lower & $n^{1-\eps}$,$\forall \eps >0$ & $n^{1-\eps}$,$\forall \eps >0$  & $n^{1-\eps}$,$\forall \eps >0$  & $n^{1-\eps}$,$\forall \eps >0$  \\
   \hline
   \multirow{2}{*}{Constant} & Upper  & $n$ & $n$ & $n$ & $n$ \\
           & Lower & $n^{1-\eps}$,$\forall \eps >0$ & $n^{1-\eps}$,$\forall \eps >0$  & $n^{1-\eps}$,$\forall \eps >0$  & $n^{1-\eps}$,$\forall \eps >0$ \mytabref{th:not-n-approx-cst} \\
   \hline
   \multirow{2}{*}{Majority} & Upper  &   $n$    & $n$ & $n$ & $n$ \\
                              & Lower  & $n^{1-\eps}$,$\forall \eps >0$  &   $n^{1-\eps}$,$\forall \eps >0$    & $n^{1-\eps}$,$\forall \eps >0$ & $n^{1-\eps}$,$\forall \eps >0$\mytabref{th:not-n-approx} \\
   \hline
   \multirow{2}{*}{Unanimity} & Upper  & $2^k$ \mytabref{th:op-2k-approx}      & $r(n), \forall r$ \mytabref{cor:op-rn-approx} & $2^k$   & $r(n), \forall r$ \\
                              & Lower  & $n^{1-\eps}$,$\forall \eps >0$ \mytabref{th:noapprox-unanimity} & ?      & $1+\varepsilon$ \mytabref{th:no-ptas-regular} & ? \\
   \hline
\end{tabular}
}
\caption{Table of the approximation results for \maxopinfluence and \maxclinfluence. A~?~symbol represents an open question.} 
\label{tab:results}
\end{table}

While MinTSS is both constant-approximable in polynomial time and fixed-parameter tractable for the unanimity case, this does not hold anymore for our problem. 
Indeed, we prove that, in this case, \maxclinfluence (resp. \maxopinfluence) is strongly inapproximable in polynomial-time and the decision version, denoted by \influence, is~$\wone$-hard w.r.t. the combined parameter $(k, \ell)$ where~$\ell$ denotes the number of vertices activated during the propagation process. 
However, we show that \maxclinfluence (resp. \maxopinfluence) becomes approximable if we are allowed to use fpt-time and \influence gets fixed-parameter tractable w.r.t combined parameter~$(k,\Delta)$, where~$\Delta$ is the maximum degree of the input graph.


Our paper is organized as follows. In \autoref{sec:preliminaries}, after  introducing some preliminaries, we establish some basic lemmas. In \autoref{sec:paraminapprox} we study   \maxopinfluence and \maxclinfluence with majority thresholds and thresholds at most two. In \autoref{sec:unanimitythresholds} we study the case of unanimity thresholds in general graphs and in bounded degree graphs.  Conclusions are provided in \autoref{sec:conclusions}.


%
%
%

\section{Preliminaries \& Basic Observations} \label{sec:preliminaries}


In this section, we provide basic backgrounds and notation used throughout this paper, give the statements of the studied problems, and establish some lemmas.

\paragraph{Graph terminology} Let $G=(V,E)$ be an \textit{undirected graph}.  For a subset $S\subseteq V$, $G[S]$ is the subgraph induced by $S$. The \textit{open neighborhood} of a vertex $v \in V$, denoted by $N(v)$, is the set of all neighbors of $v$. The \textit{closed neighborhood} of a vertex $v$, denoted $N[v]$, is the set $N(v) \cup \{v\}$.
Furthermore, for a vertex set~$V' \subset V$ we set~$N(V') = \bigcup_{v\in V'} N(v)$ and~$N[V'] = \bigcup_{v\in V'} N[v]$.
The set $N^{k}[v]$, called the $k$-neighborhood of $v$, denotes the set of vertices which are at distance at most $k$ from $v$ (thus $N^1[v]= N[v]$).
The \textit{degree} of a vertex $v$ is denoted by $\deg_G(v)$ and the \emph{maximum degree} of the graph~$G$ is denoted by~$\Delta_G$.
We skip the subscript if~$G$ is clear from the context.
Two vertices are \textit{twins} if they have the same neighborhood. They are called \textit{true twins} if they are moreover neighbors, \textit{false twins} otherwise.


\paragraph{Parameterized complexity}
A parameterized problem $(I,k)$ is said \textit{fixed-parameter tractable} (or in the class $\fpt$) w.r.t. parameter $k$ if it can be solved in $f(k)\cdot|I|^c$ time, where $f$ is any computable function and $c$ is a constant (one can see \cite{DF99,Nie06}).
The parameterized complexity hierarchy is composed of the classes $\fpt \subseteq \wone \subseteq \wtwo \subseteq \dots \subseteq \mathsf{W[P]}$.
A $\wone$-hard problem is not fixed-parameter tractable (unless $\fpt=\wone$) and one can prove $\wone$-hardness by means of a \emph{parameterized reduction} from a $\wone$-hard problem.
A parameterized reduction a mapping of an instance~$(I,k)$ of a problem~$A_1$ in $g(k)\cdot |I|^{O(1)}$ time (for any computable~$g$) into an instance $(I',k')$ for~$A_2$ such that $(I,k)\in A_1\Leftrightarrow (I',k')\in A_2$ and $k'\le h(k)$ for some~$h$.

\paragraph{Approximation} Given an  optimization problem $Q$ and an instance $I$ of this problem, we denote by  $|I|$ the size of $I$, by $opt_{Q}(I)$ the optimum value of $I$ and by $val(I,S)$  the value of a feasible solution $S$ of $I$.
For any feasible solution $S$ of $I$, we assume that $|S|$ is polynomially bounded in $|I|$ \ie $|S| \leq |I|^{O(1)}$.

\sloppy The {\em performance ratio\/} of $S$ (or {\it approximation factor}) is $r(I,S)=\max\left\{\frac{val(I,S)}{opt_{Q}(I)}, \frac{opt_{Q}(I)}{val(I,S)}\right\}.$ The  {\em error} of $S$, $\eps(I,S)$, is defined by~${\eps(I,S)= r(I,S)-1}$. For a function $f$ (resp. a constant $c>1$), an algorithm is a $f(n)$-approximation (resp. a $c$-approximation) if for any instance $I$ of $Q$ it returns a solution $S$ such that $r(I,S) \leq f(n)$ (resp. $r(I,S) \leq c$).

An optimization problem is polynomial-time \textit{constant approximable} (resp. has a \textit{polynomial-time approximation scheme}) if, for some constant $c > 1$ (resp. every constant $\varepsilon > 0$), there exists a polynomial-time $c$-approximation (resp. $(1 + \varepsilon)$-approximation) for it.

An optimization problem is \textit{$f(n)$-approximable in fpt-time w.r.t. parameter~$k$} if there exists an $f(n)$-approximation running in time $g(k)\cdot|I|^{O(1)}$, where $k$ is a given positive integer called parameter and $g$ is any computable function~\cite{marx06}. 

The notion of an $E$-reduction ({\it error-preserving} reduction)
was introduced  by Khanna \emph{et al.}~\cite{KMSV99}. A problem $Q$ is
called {\it $E$-reducible} to a problem $Q'$, if there exist
polynomial-time computable functions $f$, $g$ and a constant
$\beta$ such that
\begin{itemize}
\item $f$ maps an instance $I$ of $Q$ to an instance $I'$ of $Q'$
such that $opt(I)$ and $opt(I')$ are related by a polynomial
factor, i.e. there exists a polynomial $p(n)$ such that
$opt(I')\leq p(|I|) opt(I)$, \item $g$ maps solutions $S'$ of $I'$
to solutions $S$ of $I$ such that $\eps(I,S)\leq \beta
\eps(I',S')$.
\end{itemize}

An important property of an  $E$-reduction is that it can be applied
uniformly to all levels of approximability; that is, if $Q$ is
$E$-reducible to $Q'$ and $Q'$ belongs to $\cal{C}$ then $Q$
belongs to $\cal{C}$ as well, where $\cal{C}$ is a class of
optimization problems with any kind of approximation guarantee.

It is worth noting that the investigated problems in this paper are in fact \textit{cardinality constrained problems}. Recall that a problem of this kind asks for finding a solution of $k$ elements that optimizes an objective function~\cite{cai08}.
For such problems a natural choice for the parameter is the cardinality~$k$ of the solutions.

\paragraph{Problems definition}
Let $G=(V,E)$ be an undirected graph and let~$\thr\colon V \to \N$ be a threshold function.
In this paper, we consider majority thresholds \ie $\thr(v) = \lceil \frac{\deg(v)}{2} \rceil$ for each $v \in V$, unanimity thresholds \ie $\thr(v) = \deg(v)$ for each $v \in V$, and constant thresholds \ie $\thr(v) \leq c$ for each $v \in V$ and some constant $c > 1$.
Initially, no vertex is active and we select a subset $S \subseteq V$ of $k$ vertices.
The propagation unfolds in discrete steps.
At time step $0$, only the vertices in $S$ are activated.
At time step $t+1$, a vertex $v$ is activated if and only if the number of its activated neighbors at time $t$ is at least $\thr(v)$.
We apply the rule iteratively until no more activations are possible.
Given that $S$ is the set of initially activated vertices $\sigma[S]$ is the set of all activated vertices at the end of the propagation process and $\sigma(S)$ is the set $\sigma[S] \setminus S$.
The optimization problems we consider are then defined as follows.


\problemopt{\maxopinfluence}{A graph $G=(V,E)$, a threshold function $\thr:V \to \N$, and an integer $k$.}{A subset $S \subseteq V$, $|S| \leq k$ such that $|\sigma(S)|$ is maximum.}

Similarly, the \maxclinfluence problem asks for a set $S$ such that $|\sigma[S]|$ is maximum. The corresponding decision version \influence is also studied. 
Notice that in this case considering either~$\sigma[S]$ or~$\sigma(S)$ is equivalent.



\problemopt{\influence}{A graph $G=(V,E)$, a threshold function $\thr:V \to \N$, and two integers $k$ and $\ell$.}{Is there a subset $S \subseteq V$, $|S| \leq k$ such that $|\sigma(S)| \geq \ell$ ?}





\paragraph{Basic results}

In the following, we state and prove some lemmas that will be used later in the paper.

\begin{lemma} \label{lem:op-implies-cl} 
Let $r$ be any computable function. If \maxopinfluence is $r(n)$-approximable then \maxclinfluence is also $r(n)$-approximable where $n$ is the instance size.
\end{lemma}

\begin{proof}
Let $A$ be an $r(n)$-approximation algorithm for \maxopinfluence. Let $I$ be an instance of \maxclinfluence and $opt(I)$ its optimum value. When we apply $A$ on $I$ it returns a solution $S$
such that $|\sigma(S)| \geq \frac{opt(I)-k}{r(n)}$ and then $|\sigma[S]| = k+|\sigma(S)| \geq \frac{opt(I)}{r(n)}$.
\qedfill
\end{proof}

\begin{lemma} \label{lem:k-implies-n}
If an optimization problem is $r_1(k)$-approximable in fpt-time w.r.t. parameter~$k$ for \textbf{some} strictly increasing function $r_1$ depending solely on $k$ then it is also $r_2(n)$-approximable in fpt-time w.r.t. parameter~$k$ for \textbf{any} strictly increasing function $r_2$ depending solely on the instance size~$n$.
\end{lemma}

\begin{proof}
Let $r_1^{-1}$ and $r_2^{-1}$ be the inverse functions of $r_1$ and $r_2$, respectively. Let~$I$ be an instance of a maximization problem with size~$n = |I|$ (the proof is analogous for minimization problems). We distinguish the following two cases.

\textbf{Case 1:} $k \leq r_1^{-1}(r_2(n))$. In this case, we apply the $r_1(k)$-approximation algorithm and directly get a solution $S$ such that $r(I,S) \leq r_1(k) \leq r_1(r_1^{-1}(r_2(n))) = r_2(n)$ in time $f(k) \cdot n^{O(1)}$ for some computable function~$f$.

\textbf{Case 2:} $k > r_1^{-1}(r_2(n))$. We then have $n < r_2^{-1}(r_1(k))$ and thus we can solve the instance~$I$ by exhaustively checking every solution $S$ of $I$ and return the one with the largest $val(I,S)$ value. Since we have $|S| \leq n^{O(1)}$ (see the discussion above), we know that there are at most $2^{n^{O(1)}} \leq 2^{r_2^{-1}(r_1(k))^{O(1)}}$ different solutions assuming, without loss of generality, that the solutions are encoded in binary. It follows that the running time in this case is $2^{r_2^{-1}(r_1(k))^{O(1)}} = f(k)$ for some computable function~$f$.
This completes the proof.
\qedfill
\end{proof}

As an illustration of this lemma, if a problem admits a polynomial-time~$k$-approximation then we can approximate this problem within any arbitrarily small ratio depending on the instance size in fpt-time \textit{e.g.} $\log(\log(\ldots\log(n))$.

It is worth pointing out that a problem which is proven inapproximable in fpt-time obviously implies that it is not approximable in polynomial time with the same ratio.
Therefore, fpt-time inapproximability can be considered as a ``stronger'' result than polynomial-time inapproximability.

\section{Parameterized inapproximability}\label{sec:paraminapprox}

In this section, we consider the parameterized approximability of both \maxclinfluence and \maxopinfluence. We show that these problems are $\wtwo$-hard
to approximate within $n^{1 - \eps}$  for any $\eps >0$ for majority thresholds and thresholds at most two.
To do so, we use the following polynomial-time reduction from \ds as the starting point.
The \ds problem asks, given an undirected graph~$G=(V,E)$ and an integer~$k$, whether there is a vertex subset~$S \subseteq V$, $|S| \le k$, such that~$N[S] = V$.


\paragraph{Basic Reduction} 
Given an instance~$I=(G=(V,E), k)$ of \ds we construct the instance~$I'=(G' = (V',E'), \thr, k, |V'|)$ of \influence[k,|V'|] as follows.
For each vertex~$v \in V$, we add two vertices~$v^t$ and~$v^b$ ($t$ and $b$ respectively standing for \emph{top} and \emph{bottom}) to~$V'$ as well as the edge $v^t v^b$ to~$E'$.
For each edge~$uv \in E$, add the edges~$u^t v^b$ and~$u^b v^t$ to $E'$.
Finally, set  $\thr(v^t) = \deg_{G'}(v^t)$ and $\thr(v^b) = 1$ for every top vertex $v^t$ and every bottom vertex $v^b$, respectively. This completes the reduction (see \autoref{fig:basic_dom_set}).

We claim that $I$ is a yes-instance of \ds if and only if $I'$ is a yes-instance of \influence[k,|V'|].
For the forward direction, suppose there exists a dominating set $S \subseteq V$ in $G$ of size $k$. Consider the solution $S'\subseteq V'$   containing the corresponding top vertices. After the first step, all bottom vertices are activated since they have thresholds one and $S$ is a dominating set. Finally, after the second step, all top vertices are activated too.
For the reverse direction, suppose there is a  subset $S'\subseteq V'$ of size $k$ in $G'$ such that~$\sigma[S']=V'$. 
We can assume without loss of generality that $S'$ contains no bottom vertex. 
Since all bottom vertices are activated we have that~$\{v_i:v_i^t\in S'\}$ is a dominating set in~$G$.

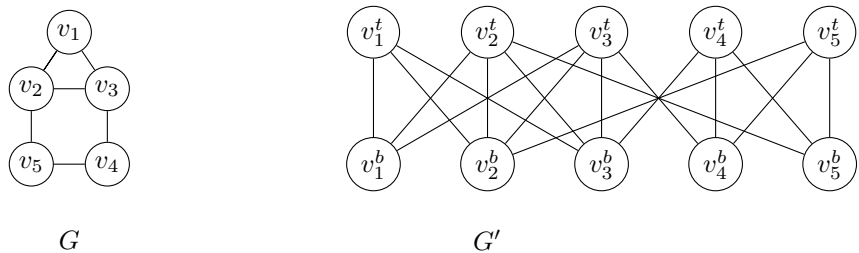
\begin{figure}
\centering
\begin{tikzpicture}[scale=1,transform shape]
\node[vertex,draw] (v5) at (0,0) {$v_5$};
\node[vertex,draw] (v4) at (1,0) {$v_4$};
\node[vertex,draw] (v3) at (1,1) {$v_3$};
\node[vertex,draw] (v2) at (0,1) {$v_2$};
\node[vertex,draw] (v1) at (0.5,1.75) {$v_1$};
\draw (v1) -- (v2) -- (v3) -- (v4) -- (v5) -- (v2) -- (v1) -- (v3);
\node () at (0.5,-1) {$G$};
\begin{scope}[xshift=3cm]
\foreach \x in {1,...,5} {
	\node[vertex,draw] (vt\x) at (1.5*\x,1.75) {$v_ {\x}^t$};
	\node[vertex,draw] (vb\x) at (1.5*\x,0) {$v_ {\x}^b$};
	\draw (vt\x) -- (vb\x);
}
\draw (vt1) -- (vb2);
\draw (vt2) -- (vb1);
\draw (vt1) -- (vb3);
\draw (vt3) -- (vb1);
\draw (vt2) -- (vb3);
\draw (vt3) -- (vb2);
\draw (vt4) -- (vb3);
\draw (vt3) -- (vb4);
\draw (vt4) -- (vb5);
\draw (vt5) -- (vb4);
\draw (vt2) -- (vb5);
\draw (vt5) -- (vb2);
\node () at (3,-1) {$G'$};
\end{scope}
\end{tikzpicture}
\caption{Sample construction of the bipartite graph $G'$  from a graph $G$ of \ds. All vertices $v_i^t$, $1 \leq i \leq 5$ have thresholds $\deg_{G'}(v_i^t)$ while all vertices $v_i^b$, $1 \leq i \leq 5$ have thresholds 1.}\label{fig:basic_dom_set}
\end{figure}

\paragraph{Inapproximability results} We are now ready to prove the main results of this section.

\noindent
\begin{theorem} \label{th:not-n-approx}
For any $\eps >0$, \maxclinfluence and \maxopinfluence with majority thresholds cannot be approximated within $n^{1-\eps}$ in fpt-time w.r.t. parameter~$k$ even on bipartite graphs, unless $\fpt=\wtwo$.
\end{theorem}

\begin{proof}
By \autoref{lem:op-implies-cl}, it suffices to show the result for \maxclinfluence. 
We  provide a polynomial-time reduction from \ds to \maxclkplusuninfluence with majority thresholds.
In this reduction, we will make use of the $q$-\emph{edge} gadget, for some integer $q$. An $q$-edge between two vertices $u$ and $v$
consists of $q$ vertices of threshold one adjacent to both $u$ and $v$. 




Given an instance $I=(G=(V,E),k)$ of \ds with $n = |V|$, $m=|E|$, we define an  instance $I'$ of \maxclkplusuninfluence. We start with the \emph{basic reduction} and modify $G'$ and the function $\thr$ as follows.
Replace every edge $v^tv^b$ by an $(k+2)$-edge between $v^t$ and $v^b$. Moreover, for a given constant~$\beta = \frac{8}{\eps}-5$, let $L = \lceil n^{\beta} \rceil$ and we add $nL$ more vertices $x_{1}^1, \ldots, x_{n}^1, \ldots, x_{1}^L, \ldots, x_{n}^L$.  For $i=1,\ldots,n$, vertex  $x_{i}^1$ is adjacent to all the bottom vertices. Moreover, for any $j=2,\ldots,L$, each $x_{i}^j$ is adjacent to  $x_{k}^{j-1}$, for any $i,k\in\{1,\ldots,n\}$.   We also add a vertex $w$ and an $n+(k+2)(\deg_G(v)-1)$-edge between $w$ and $v^b$, for any bottom vertex $v^b$. For $i=1,\ldots,n$, vertex $x_{i}^1$ is adjacent to  $w$.
For  $i=1,\ldots,n$ add $n$ pending-vertices (\ie degree one vertices) adjacent to $x_{i}^L$. For any vertex
$v^t$ add $(\deg_G(v)+1)(k+2)$ pending-vertices adjacent to $v^t$. Add also $n+n^2+(k+2)(2m-n)$ pending-vertices adjacent to $w$.
All vertices of the graph $G'$ have the majority thresholds (see also \autoref{fig:reduc2}).


\begin{figure}[t]
	\begin{center}
		\begin{tikzpicture}[scale=1,transform shape]
			\begin{scope}[yshift=-3cm]
				\node[vertex,draw] (v5) at (0,0) {$v_5$};
				\node[vertex,draw] (v4) at (1,0) {$v_4$};
				\node[vertex,draw] (v3) at (1,1) {$v_3$};
				\node[vertex,draw] (v2) at (0,1) {$v_2$};
				\node[vertex,draw] (v1) at (0.5,1.75) {$v_1$};
				\draw (v1) -- (v2) -- (v3) -- (v4) -- (v5) -- (v2) -- (v1) -- (v3);
			\end{scope}
			\begin{scope}[xshift=4cm]
				\node[vertex,draw,star,star points=30,star point ratio=0.7] (u) at (0,-1.25) {$w$};

				\foreach \x in {1,...,5} {
					\node[vertex,draw,star,star points=30,star point ratio=0.7] (vt\x) at (\x*1.5-0.5,1.75) {$v_ {\x}^t$};

				  \ifthenelse{\x = 4}{
					  \node[vertex,draw] (v0-\x) at (\x*1.5-0.5,0) {$v_ {\x}^b$};
				  }{
				    \ifthenelse{\x = 5}{
					    \node[vertex,draw] (v0-\x) at (\x*1.5-0.5,0) {$v_ {\x}^b$};
            }{
					    \node[vertex,draw] (v0-\x) at (\x*1.5-0.5,0) {$v_ {\x}^b$};
            }
				  }

					\draw[ultra thick] (vt\x) -- (v0-\x);
          \draw[ultra thick, bend right=5] (u) to (v0-\x);
				}
				\foreach \x / \y in {1/2, 1/3, 2/3, 2/5, 3/4, 4/5} {
					\draw[ultra thick] (vt\x) -- (v0-\y);
					\draw[ultra thick] (vt\y) -- (v0-\x);
				}
				\foreach \i in {1,...,4} {
					\foreach \x in {1,...,5} {
						\ifthenelse{\i = 3}{
							\node[] (v\i-\x) at (\x*1.5-0.5,-1.75*\i) {$\vdots$};
						}{
						  \ifthenelse{\i = 4}{
                \node[vertex,draw,star,star points=30,star point ratio=0.7] (v\i-\x) at (\x*1.5-0.5,-1.75*\i) {$x_ {\x}^{L}$};
						  }{

				        \ifthenelse{\x = 4}{
							      \node[vertex,draw] (v\i-\x) at (\x*1.5-0.5,-1.75*\i) {$x_ {\x}^{\i}$};
				        }{
				          \ifthenelse{\x = 5}{
							      \node[vertex,draw] (v\i-\x) at (\x*1.5-0.5,-1.75*\i) {$x_ {\x}^{\i}$};
                  }{
							      \node[vertex,draw] (v\i-\x) at (\x*1.5-0.5,-1.75*\i) {$x_ {\x}^{\i}$};
                  }
				        }
						  }
						}
					}
				}

       \draw (u) to (v1-1);
       \draw[bend left=15] (u) to (v1-2);
       \draw[bend left=15] (u) to (v1-3);
       \draw[bend left=15] (u) to (v1-4);
       \draw[bend left=15] (u) to (v1-5);



			\begin{pgfonlayer}{bg}
				\foreach \i in {1,...,4} {
					\foreach \x in {1,...,5} {
						\pgfmathtruncatemacro{\value}{\i - 1};	\let\j=\value
						\foreach \y in {1,...,5} {
							\draw[draw=black] (v\j-\x) -- (v\i-\y);
						}
					}
				}
		  \end{pgfonlayer}
			\end{scope}
		\end{tikzpicture}
	\end{center}
  \caption{The graph $G'$ (right) obtained from $G$ (left) after carrying out the modifications of \autoref{th:not-n-approx}. A thick edge represents an $q$-edge for some $q > 0$. A ``star'' vertex $v$ represents a vertex adjacent to $\frac{\deg_{G'}(v)}{2}$ pending-vertices.}
  \label{fig:reduc2}
\end{figure}
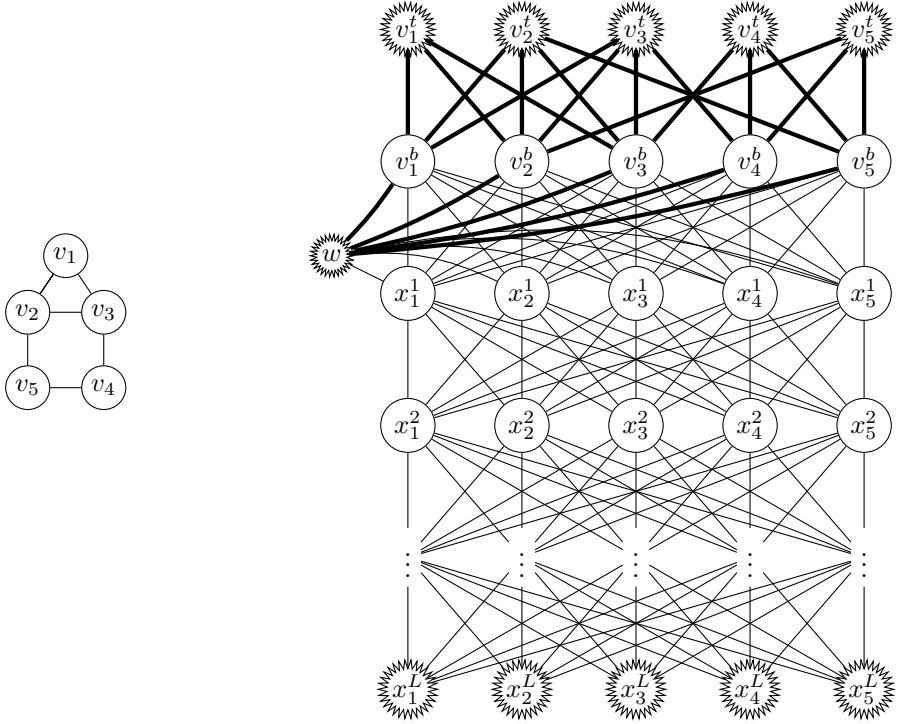

We claim that if $I$ is a \textit{yes}-instance then $opt(I') \geq n L\geq n^{\beta+1}$; otherwise $opt(I') < n^4$. Let $n' = |V'|$, notice that we have $n' \leq  n^4 + nL$.

Suppose that there exists a dominating set $S \subseteq V$ in $G$ of size at most $k$. Consider the solution $S'$ for $I'$ containing the corresponding top vertices and  vertex $w$. After the first round, all vertices belonging to the edge gadgets which top vertex is in $S'$ are activated. Since $S$ is a dominating set in $G$, after the second round,  all the bottom vertices are activated.
Indeed $\deg_{G'}(v^b)=2(n+(k+2)\deg_G(v))$ and after the first round $v^b$ has at least $k+2$ neighbors activated belonging to an $(k+2)$-edge between $v^b$ and  some $u^t\in V$ and $n+(k+2)(\deg_G(v)-1)$ neighbors activated belonging to an $n+(k+2)(\deg_G(v)-1)$-edge between $v^b$ and $w$.
 Thus, every vertex $x_{i}^1$ gets active after the third round, and generally after the $j$th  round, $j=4,\ldots,L+2$ the vertices $x_i^{j-2}$ are activated, and at the $(L+3)$th round all pending-vertices adjacent to $x_i^L$ are activated. Therefore, the size of an optimal solution is at least $nL\geq n^{\beta+1}$.

Suppose that there is no dominating set in $G$ of size  $k$. Without loss of generality, we may assume that no pending-vertices are in a solution of $I'$ since they all have threshold one. If $w$ does not take part of a solution in $I'$, then no vertex $x_i^1$ could be activated and in this case $opt(I')$ is less than $n' - nL \leq n^4$. Consider now the solutions of $I'$ of size $k+1$ that contain $w$. Observe that if a top-vertex $v^t$ gets active through bottom-vertices then $v^t$ can not activate any other bottom-vertices.
Indeed, as a contradiction, suppose that $v^t$ is adjacent to a non-activated bottom-vertex. It follows that $v^t$ could not have been activated because of its threshold and that no pending-vertices are part of the solution, a contradiction.
Notice also that it is not possible to activate a bottom vertex by selecting some $x_i^1$ vertices since of their threshold.
Moreover, since there is no dominating set of size $k$, any subset of $k$ top vertices cannot activate all bottom vertices, therefore no vertex $x_i^k$, $i=1,\ldots,n, k=1,\ldots, L$ can be activated. Hence, less than $n' - nL$ vertices can be activated in $G'$ and the size of an optimal solution is at most $n^4$.

Assume now that there is an fpt-time $n^{1-\eps}$-approximation algorithm $A$ for \maxclkplusuninfluence with majority threshold. Thus, if $I$ is a \textit{yes}-instance, the algorithm gives a solution of value $A(I') \geq \frac{n^{\beta+1}}{(n')^{1-\eps}} > \frac{n^{\beta+1}}{n^{(1-\eps)(\beta+5)}} = n^{4}$ since $n' \leq  n^4 + nL < n^5L$.
If $I$ is a \textit{no}-instance, the solution value is $A(I') < n^4$. Hence, the approximation algorithm $A$ can distinguish in fpt-time between \textit{yes}-instances and \textit{no}-instances for \ds implying that $\fpt=\wtwo$ since this last problem is $\wtwo$-hard~\cite{DF99}.
\qedfill
\end{proof}

\begin{theorem} \label{th:not-n-approx-cst}
For any $\eps \geq 0$, \maxclinfluence and \maxopinfluence with thresholds at most two cannot be approximated within $n^{1 - \eps}$ in fpt-time w.r.t. parameter~$k$ even on bipartite graphs, unless $\fpt=\wtwo$.
\end{theorem}

{
\begin{proof}
By \autoref{lem:op-implies-cl}, it suffices to prove the result for \maxclinfluence. We  construct a polynomial-time reduction from \ds to \maxclinfluence with thresholds at most two.
In this reduction, we will make use of the \emph{directed edge} gadget. A directed edge from a vertex $u$ to another vertex $v$ consists of a $4$-cycle $\{a,b,c,d\}$ such that $a$ and $u$ as well as~$c$ and $v$ are adjacent. Moreover $\thr(a) = \thr(b) = \thr(d) = 1$ and $\thr(c) = 2$.
The idea is that the vertices in the directed edge gadget become active if~$u$ is activated but not if~$v$ is activated.
Hence, the activation process may go from~$u$ to~$v$ via the gadget but not in the reverse direction.
In the rest of the proof, we may assume that no vertices from $\{a,b,c,d\}$ are part of a solution of \maxclinfluence.
Indeed, it is always as good to take the vertex $u$ instead.
We will also make use of a \emph{directed tree} with leaves $x_1,\ldots,x_n$ and root $r$ defined as follows: introduce $n-1$ new vertices $y_2,\ldots,y_n$ and insert a directed edge from $x_1$ to $y_2$, from $x_2$ to $y_2$, from $y_i$ to $y_{i+1}$ for $i=2,\ldots,n-1$, from $x_i$ to $y_i$ for $i=3,\ldots,n$, and from $y_n$ to $r$.
Moreover $\thr(y_i)=2$, $i=2,\ldots,n$ and $\thr(r)=1$. The idea is that the vertices in the directed tree become active if all vertices~$x_1,\ldots,x_n$ are activated but not if~$r$ is activated.
So, we may assume that no vertex from $y_2,\ldots,y_n$ is part of a solution of \maxclinfluence.

Given an instance $I=(G=(V,E),k)$ of \ds with $n = |V|$, we define an  instance $I'$ of \maxclinfluence. We start with the \emph{basic reduction} and modify $G'$ and the function $\thr$ as follows. Set the thresholds of top-vertices to two. Replace every edge between a top vertex $v^t$ and a bottom vertex $v^b$ by a directed edge from $v^t$ to $v^b$.
For $j=1,\ldots,n^{\beta}$, where~$\beta= \frac{4}{\eps} - 3$, add vertices $p_{1}^{j}, \ldots, p_{n}^{j}$ and a directed tree between leaves $v_i^b$, $i=1,\ldots,n$ and root  $p_{\ell}^{1}$, for $\ell=1,\ldots,n$. Moreover for $j=1,\ldots,n^{\beta}-1$ add directed trees between leaves $p_{1}^{j}, \ldots, p_{n}^{j}$ and root  $p_{\ell}^{j+1}$, for $\ell=1,\ldots,n$.
 This completes the construction (see \autoref{fig:not-n-approx-cst}).
Let $n' = |V'|$, notice that we have $n' = 2n + n^2n^{\beta} +4(2n-1)n^{\beta + 1}< n^{\beta + 3}$.

\begin{figure}[t]
	\begin{center}
		\begin{tikzpicture}[scale=1,transform shape]
			\begin{scope}[yshift=-3cm]
				\node[vertex,draw] (v5) at (0,0) {$v_5$};
				\node[vertex,draw] (v4) at (1,0) {$v_4$};
				\node[vertex,draw] (v3) at (1,1) {$v_3$};
				\node[vertex,draw] (v2) at (0,1) {$v_2$};
				\node[vertex,draw] (v1) at (0.5,1.75) {$v_1$};
				\draw (v1) -- (v2) -- (v3) -- (v4) -- (v5) -- (v2) -- (v1) -- (v3);
			\end{scope}
			\begin{scope}[xshift=4cm]
				\foreach \x in {1,...,5} {
					\node[vertex,draw] (vt\x) at (\x*1.5-0.5,1.75) {$v_ {\x}^t$};
					\node[vertex,draw] (v0-\x) at (\x*1.5-0.5,0) {$v_ {\x}^b$};
					\draw[thick,->] (vt\x) -- (v0-\x);
				}
				\foreach \x / \y in {1/2, 1/3, 2/3, 2/5, 3/4, 4/5} {
					\draw[thick,->] (vt\x) -- (v0-\y);
					\draw[thick,->] (vt\y) -- (v0-\x);
				}
				\foreach \i in {1,...,4} {
					\foreach \x in {1,...,5} {
						\ifthenelse{\i = 3}{
							\node[] (v\i-\x) at (\x*1.5-0.5,-1.75*\i) {$\vdots$};
						}{
						  \ifthenelse{\i = 4}{
							\node[vertex,draw] (v\i-\x) at (\x*1.5-0.5,-1.75*\i) {$p_ {\x}^{L}$};
						  }{
							  \node[vertex,draw] (v\i-\x) at (\x*1.5-0.5,-1.75*\i) {$p_ {\x}^{\i}$};
						  }
						}
					}
				}

			\begin{pgfonlayer}{bg}
				\foreach \i in {1,...,4} {
					\foreach \x in {1,...,5} {
						\pgfmathtruncatemacro{\value}{\i - 1};	\let\j=\value
						\foreach \y in {1,...,5} {
							\draw[draw=black!30,->] (v\j-\x) -- (v\i-\y);
						}
					}
				}
		  \end{pgfonlayer}
			\end{scope}
		\end{tikzpicture}
	\end{center}
	\caption{The graph $G'$ (right) obtained from $G$ (left) after carrying out the modifications of \autoref{th:not-n-approx-cst}. A black arrow from $u$ to $v$ represents a directed edge gadget from $u$ to $v$. A gray arrow from $u$ to $v$ indicates a directed tree where~$u$ is one of the leafs and~$v$ is the root.}
	\label{fig:not-n-approx-cst}
\end{figure}
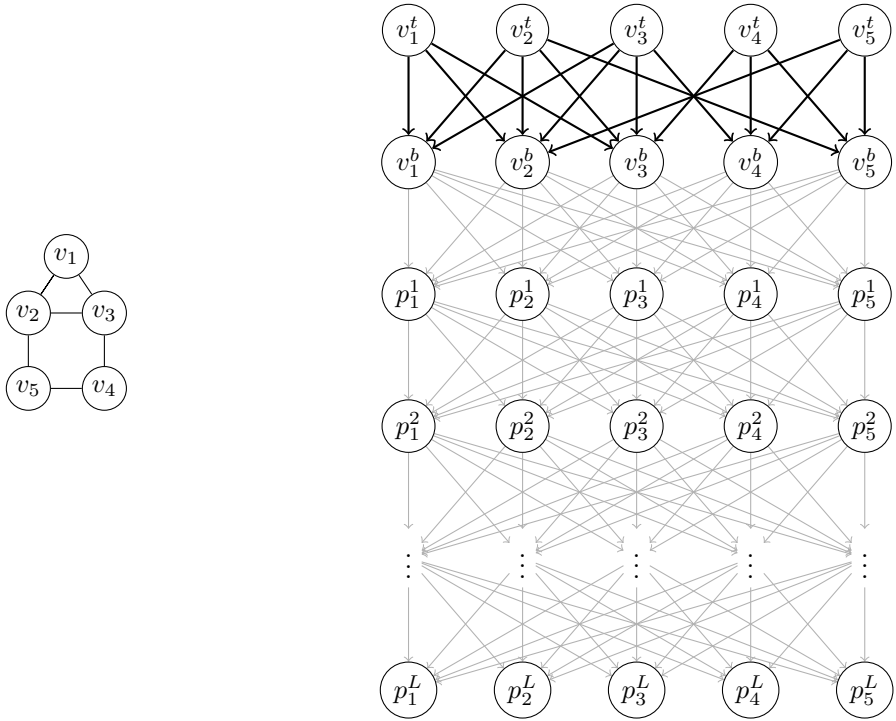

We claim that if $I$ is a \textit{yes}-instance then $opt(I') > n^{\beta + 2}$; otherwise $opt(I') < n^{3}$.

Suppose that there exists a dominating set $S \subseteq V$ in $G$ of size at most $k$. Consider the solution $S'$ for $I'$ containing the corresponding top vertices. Since $S$ is a dominating set in $G$, after the fourth round,  all the bottom vertices are activated.  It follows that at the end of the activation process all the vertices of the graph $G'$ are activated except the top vertices outside $S'$ and the vertices of some directed edges of the basic gadget. The optimum solution is $opt(I') > n'-5n^2 >n^{\beta + 2}$.

Suppose that there is no dominating set in $G$ of size  $k$. Consider a solution $S'$ for $I'$ of size $k$. Without loss of generality, we may assume that no  $p_{i}^{j}$ vertices or bottom vertices are contained in $S'$ since they all have threshold one. For the reason previously mentioned, we know that no vertices from the directed edge gadgets and no vertices from the directed trees are in $S'$. It follows that $S'$ only contains top-vertices. Since there is no dominating set of size $k$ in $G$ then at least one bottom-vertex is not activated.
Moreover, because of the directed edges the activated bottom-vertices cannot activate new top-vertices. Thus at least one vertex of each directed tree  with roots $p_i^1$, $i=1,\ldots,n$ cannot be activated implying that no  $p_{i}^{j}$ vertices can be activated. This leads to a solution of size at most $5n^2 < n^{3}$.

Assume now that there is an 
fpt-time $n^{1 - \eps}$-approximation algorithm $A$ for \maxclinfluence with threshold at most two. Thus, if $I$ is a \textit{yes}-instance, the algorithm gives a solution of value $A(I') \geq \frac{n^{\beta+2}}{(n')^{1-\eps}} > \frac{n^{\beta+2}}{n^{(1-\eps)(\beta+3)}} > n^{3}$ since $n' < n^{\beta +3}$.
If $I$ is a \textit{no}-instance, the solution value is $A(I') < n^{3}$. Hence, the approximation algorithm $A$ can distinguish in fpt-time between \textit{yes}-instances and \textit{no}-instances for \ds implying that $\fpt=\wtwo$ since this last problem is $\wtwo$-hard~\cite{DF99}.
\qedfill
\end{proof}
}

Using \autoref{lem:k-implies-n}, \autoref{th:not-n-approx}, and \autoref{th:not-n-approx-cst} we can deduce the following corollary.

\begin{corollary}
\sloppy For any strictly increasing function $r$, 	\maxclinfluence and  \maxopinfluence with thresholds at most two or majority thresholds cannot be approximated within~$r(k)$ in fpt-time w.r.t. parameter~$k$ unless $\fpt = \wtwo$.
\end{corollary}

\section{Unanimity thresholds} \label{sec:unanimitythresholds}

In the previous section, we proved that the problem is parameterized inapproximable even for constant and majority thresholds. In this section, we show that assuming unanimity thresholds leads to more positive results. More precisely, we give a parameterized approximation algorithm on general graphs, and show that the problem is fixed-parameter tractable w.r.t.~$k$ for the class of graphs of bounded maximum degree.



\subsection{General graphs}

We first show that, in the unanimity case, \influence is $\wone$-hard w.r.t. parameter $k+\ell$ and \maxopinfluence is not approximable within $n^{1-\eps}$ for any $\eps >0$ in polynomial time, unless $\np=\zpp$.
However, if we are allowed to use fpt-time then \maxopinfluence with unanimity is $r(n)$-approximable in fpt-time w.r.t. parameter~$k$ for any strictly increasing function $r$.

\begin{theorem}\label{th:whardness-general}
\influence with unanimity thresholds is $\wone$-hard w.r.t. the combined parameter $(k,\ell)$ even for bipartite graphs.
\end{theorem}

{
\begin{proof}
We provide a parameterized reduction from the $\wone$-hard \clique problem \cite{DF99} to \influence.
Given an instance $(G=(V,E),k)$ of \clique, we construct an instance $(G'=(V',E'),k, \ell)$ of \influence as follows. For each vertex $v \in V$ add a copy $v'$ to $V'$. For each edge $uv \in E$, add $k+1$ edge-vertices $e_{uv}^1, \ldots, e_{uv}^{k+1}$ adjacent to both $u'$ and $v'$. Set $\ell = (k+1){k \choose 2}$ and $\thr(u) = \deg_{G'}(u)$ for all $u \in V'$ (see also \autoref{fig:clique}).

\begin{figure}
\centering
\begin{tikzpicture}[minimum size=0.7cm, scale=1]
\node[vertex] (v1) at (0,0) {$v_1$};
\node[vertex] (v2) at (-1,-0.5) {$v_2$};
\node[vertex] (v3) at (0,-1) {$v_3$};
\draw (v1) -- (v2) -- (v3);
\node () at (-0.5,-2) {$G$};
\begin{scope}[xshift=3cm]
\node[vertex] (v'1) at (0,0) {$v'_1$};
\node[vertex] (v'2) at (0,-1) {$v'_2$};
\node[vertex] (v'3) at (0,-2) {$v'_3$};
\foreach \k in {1,2,3} {
	\node[vertex] (e12\k) at (1+\k,-0.1) {$e_{12}^{\k}$};
	\node[vertex] (e23\k) at (1+\k,-1.9) {$e_{23}^{\k}$};
}
	\draw (v'1) -- (e121) -- (v'2);
	\draw (v'1) edge[bend left=24] (e122) (e122) edge[bend left=9] (v'2);
	\draw (v'1) edge[bend left=24] (e123) (e123) edge[bend left=12] (v'2);
	\draw (v'2) -- (e231) -- (v'3);
	\draw (v'2) edge[bend left=8] (e232) (e232) edge[bend left=24] (v'3);
	\draw (v'2) edge[bend left=11] (e233) (e233) edge[bend left=27] (v'3);

\node () at (2,-3) {$G'$};
\end{scope}
\end{tikzpicture}
\caption{Illustration of the reduction from an instance $(G,k)$ of \clique to an instance $(G',k,\ell)$ of \influence, where $k=2$ and $\ell = 3$. }\label{fig:clique}
\end{figure}
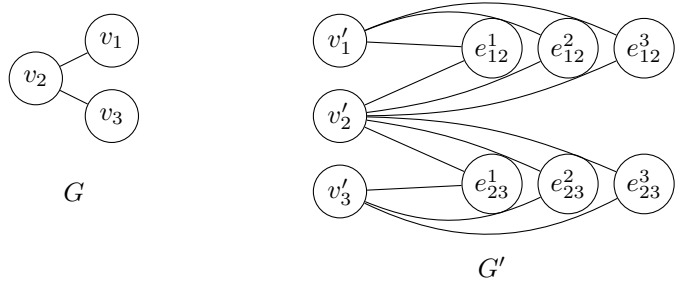

We claim that there is a clique of size $k$ in $G$ if and only if there exists a subset $S \subseteq V'$ of size $k$ such that $|\sigma(S)| \geq \ell$.


	``$\Rightarrow$'': Assume that there is a clique $C \subseteq V$ of size $k$ in $G$.  One can easily verify that the set $S = \{v' \in V': v \in C\}$ activates $|\sigma(S)| \geq (k+1){k \choose 2} = \ell$ edge-vertices in $G'$ since $C$ is clique.
	
	``$\Leftarrow$'': Suppose that there exists a subset $S \subseteq V'$ of size $k$ such that $|\sigma(S)| \geq \ell$. We may assume without loss of generality that no edge-vertices belong to $S$. Indeed, each edge-vertex is adjacent to only vertices with threshold at least $k+1$. Thus choosing some edge-vertices to $S$ cannot activate any new vertices in $G'$. Since the solution $S$ activates at least $(k+1){k \choose 2}$ edge-vertices, this implies that $S$ is a clique in $G$.
\qedfill
\end{proof}
}

\begin{theorem} \label{th:noapprox-unanimity}
For any $\eps >0$, \maxopinfluence with unanimity thresholds cannot be approximated within $n^{1-\eps}$ in polynomial time, unless $\np=\zpp$.
\label{th: }
\end{theorem}
{
\begin{proof}
We will show how to transform any approximation algorithm for \maxopinfluence into another one with the same ratio for \maxis. 
Consider the instance $I_k$ of \maxopinfluence consisting of a graph $G=(V,E)$, an integer $k$ and unanimity threshold. One can note and easily check that the following holds. Given a solution $S \subseteq V$ of
$I_k$, $\sigma(S)$ is obtained in only one step of the diffusion process and is an independent set.
Therefore there exists an integer $k^{*} \in [1,n]$ such that $\sigma(OPT(I_{k^{*}}))$ is the maximum independent set in $G$, where $OPT(I_{k^{*}})$ is the optimal solution of $I_{k^{*}}$.

Suppose that \maxopinfluence has an $f(n)$-approximation algorithm $A$, we then have
$|\sigma(A(I_{k^{*}}))| \geq \frac{|\sigma(OPT(I_{k^{*}}))|}{f(n)}$, where $A(I_{k^{*}})$ is a solution given by $A$ for the instance $I_{k^{*}}$.
It follows from the previous observation that $\sigma(A(I_{k^{*}}))$ is an independent set in $G$ and an $f(n)$-approximate solution.

Now, it suffices to apply the approximation algorithm $A$ for each $k=1,\ldots,n$
and return the approximate solution $S_{\Smax}$ that has the largest value. Given this solution, we have $|\sigma(S_{\Smax})| \geq  |\sigma(A(I_{k^{*}}))|$.
Hence, we get a polynomial-time $f(n)$-approximation algorithm for \maxis problem. Since \maxis cannot be approximated within $n^{1-\eps}$ for any $\eps >0 $ unless $\np=\zpp$~\cite{hastad96}, the result follows.
\qedfill
\end{proof}
}


In what follows, we provide an fpt-time $r(n)$-approximation algorithm w.r.t. parameter~$k$ for any strictly increasing function~$r$. As a first step toward this goal, we need the following result.

\begin{theorem} \label{th:op-2k-approx}
\maxopinfluence  and \maxclinfluence   with unanimity thresholds are \sloppy \mbox{$2^{k}$-approximable} in polynomial time.
\end{theorem}

{
\begin{proof} By \autoref{lem:op-implies-cl}, it suffices to show the result for \maxopinfluence.
The polynomial-time algorithm consists in the following two steps: (i)~Find $F$, the largest ``false-twins set'' such that $\deg(v) \leq k$, $\forall v \in F$, and (ii)~Return $N(F)$.
The first step can be done for example  by searching for the largest set of identical lines with at most $k$ ones in the adjacency matrix of the graph. Since $F$ is a false-twins set with vertices of degree at most $k$, the size of the neighborhood of $F$ is also bounded by $k$. Consider the activation of the set $N(F)$. After one round, this will activate $|\sigma(N(F))| \geq |F|$ vertices, since all the neighborhood of the vertices in $F$ are activated. 

To complete the proof, observe that for any solution of size at most $k$, there is at most $2^k$ different false-twins sets. Therefore, any optimal solution could activate at most $2^{k}\cdot |F|$ vertices, providing the claimed approximation ratio.
\qedfill
\end{proof}
}

Using \autoref{lem:k-implies-n} and \autoref{th:op-2k-approx}  we directly get the following.

\begin{corollary} \label{cor:op-rn-approx}
For any strictly increasing function $r$, \maxopinfluence and \maxclinfluence with unanimity thresholds are $r(n)$-approximable in fpt-time w.r.t. parameter~$k$.
\end{corollary}

For example, \maxopinfluence is $\log(n)$-approximable in time $O^{*}(2^{k 2^k})$, where the~$O^*$ notation suppresses polynomial factors.

\paragraph{Finding dense subgraphs} In the following we show that \maxopinfluence with unanimity thresholds is at least as difficult to approximate as the \densestksubgraph   problem, that consists of finding  in a graph a subset of vertices of cardinality $k$ that induces a maximum number of edges. In particular, any positive approximation result for \maxopinfluence with unanimity would directly transfers to \densestksubgraph. 
This last problem has no polynomial-time approximation scheme unless $\np$ has no subexponential-time algorithms~\cite{Kho06} and is $O(n^{\frac{1}{4}+\epsilon})$-approximable in time $n^{O(\frac{1}{\epsilon})}$ where $n$ is the size of the input graph~\cite{BCCFV10}. 

%
%
%
%

\begin{theorem} \label{th:DkS-fpt-approx}
For any strictly increasing function $r$,  if \maxopinfluence with unanimity thresholds is $r(n)$-approximable in fpt-time w.r.t. parameter~$k$ then  \densestksubgraph is $r(n)$-approximable in fpt-time w.r.t. parameter~$k$. 
\end{theorem}


\begin{proof}
We give an $E$-reduction from \densestksubgraph to \maxopinfluence. 
Consider an instance $I$ of \densestksubgraph  formed by a graph $G=(V,E)$ and we construct an instance $I'$ of \maxopinfluence with unanimity thresholds consisting of graph $G'=(V', E')$ as follows: for each vertex $v\in V$ add a copy $v'$ to $V'$; for each edge $uv \in E$ add an edge-vertex $e_{uv}$ to $V'$; moreover add $k+1$ vertices $x_1,\ldots,x_{k+1}$. For any edge $uv \in E$ add edges  $u'e_{uv}, e_{uv}v'$ to $E'$. Furthermore, add an edge between $x_i$ and $v'$ for any $1\leq i\leq k+1$ and any $v' \in V'$. 
Therefore, every vertex $x_i$ has degree $|V|$, every vertex $v'$ has degree $deg_G(v)+k+1$ and every edge-vertex $e$ has degree $2$.

Let  $S\subseteq V$, $|S|=k$ be an optimum solution for $I$ that is  $opt(I)$ is the number of edges induced by $S$. The set  $S'=\{v' :v \in S\}$ is such that $|\sigma(S')|=opt(I)$ since no $x$ vertex will be activated. Thus $opt(I')\geq opt(I)$.

Given any solution $S'\subseteq V'$ of size $k$, we can consider that $S'$ contains only vertices of type $v'$ such that $v\in V$.
Indeed, observe that no $v'$ and $x_i$ vertices are activated by propagation because their threshold is greater than~$k$ and there is only one step of propagation. So only edge-vertices can be activated by propagation. Therefore, it is more interesting to consider only solutions containing $v'$ vertices.  
Thus the set $S=\{v : v' \in S'\}$ has value $val(S)=val(S')$. Moreover if $S'$ is optimal, then $opt(I)\geq opt(I')$ and thus $opt(I)= opt(I')$. Therefore, we have $\varepsilon(I,S)=\varepsilon(I',S')$.
\qedfill
\end{proof}

Using \autoref{th:DkS-fpt-approx} and \autoref{cor:op-rn-approx}, we have the following corollary, independently established in \cite{Bourgeois2013}.

\begin{corollary}
	For any strictly increasing function $r$,  \densestksubgraph  is $r(n)$-approximable in fpt-time w.r.t. parameter~$k$.
\end{corollary}

\subsection{Bounded degree graphs and regular graphs}
While \maxopinfluence and \maxclinfluence are not at all approximable in polynomial time on general graphs, we show in the following that they are constant approximable in polynomial time on bounded degree graphs. 
Moreover, \maxclinfluence  and then \maxopinfluence have no polynomial-time approximation scheme even on 3-regular graphs if $\p \neq \np$. 
From the parameterized complexity point of view, we show that \influence becomes fixed-parameter tractable w.r.t. parameter $k$ on bounded degree graphs.


\begin{lemma} \label{lem:boundedDegAffect}
	\maxopinfluence  and  \maxclinfluence  with unanimity thresholds on bounded degree graphs are constant approximable in polynomial time.
\end{lemma}

{
\begin{proof} By \autoref{lem:op-implies-cl}, it suffices to show the result for \maxopinfluence.	Indeed on graphs of degree bounded by $\Delta$, the optimum is bounded by $k \cdot \Delta$ and we can construct in polynomial time a solution $S$ of value at least $\lfloor \frac{k}{\Delta} \rfloor$ by considering iteratively vertices with disjoint neighborhoods  and putting their neighbors  in $S$.
\qedfill
\end{proof}
}

\begin{theorem} \label{th:no-ptas-regular}
\maxopinfluence  and  \maxclinfluence with unanimity thresholds have no polynomial-time approximation scheme even on 3-regular graphs for $k=\theta(n)$, unless $\p=\np$.
\end{theorem}

{
\begin{proof} 
By \autoref{lem:op-implies-cl}, it suffices to show the result for \maxclinfluence. We show that if \maxclinfluence with unanimity thresholds has a polynomial-time approximation scheme $A_{\varepsilon'}$, $\varepsilon' \in(0,1)$,
on 3-regular graphs when $k=\theta(n)$, then \minvc
has also a polynomial-time approximation scheme on 3-regular graphs. Consider $G=(V,E)$ a 3-regular graph.  Clearly, a minimum vertex cover has a value $opt(G)$ satisfying  $\frac{n}{2}\leq opt(G)<n$.
For any $\varepsilon \in(0,1)$, we apply the polynomial-time approximation scheme  $A_{\varepsilon'}$ that establishes  an $(1+\varepsilon')$-approximation for  \maxclinfluence on graph $G$ for each $k$ between $\frac{n}{2}$ and $n$ and $\varepsilon'=\frac{\varepsilon}{2-\varepsilon}$.
By applying  $A_{\varepsilon'}$ on $G$ for $k$ between $\frac{n}{2}$ and $n$, we obtain a solution  $S_k\subset V$ of size $k$ such that $S_k\cup \sigma(S_k)$ is an $(1+\varepsilon')$-approximation. The set $V\setminus \sigma(S_k)$  is a vertex cover in $G$ of size denoted by $val_k$. We show in the following that the best solution obtained in this way is an $(1+\varepsilon)$-approximation for \minvc on $G$. Indeed the best solution obtained in this way has a value $val^*\leq val_{\ell}$, where $val_{\ell}$ is the value of the solution obtained for $\ell=opt(G)$. Thus $val_{\ell}=|V\setminus \sigma(S_{\ell})|$. Since $|S_{\ell}\cup \sigma(S_{\ell})|$ is an $(1+\varepsilon')$-approximation and the optimum solution activates all vertices,  we have $|S_{\ell}\cup \sigma(S_{\ell})|\geq \frac{n}{1+\varepsilon'}$ and   $|V\setminus (S_{\ell}\cup\sigma(S_{\ell}))|\leq n\frac{\varepsilon'}{1+\varepsilon'}$. Thus $val^* \leq val_{\ell}\leq \ell+n\frac{\varepsilon'}{1+\varepsilon'}\leq \ell(1+\frac{2\varepsilon'}{1+\
\varepsilon'})=\ell(1+ \varepsilon)$.
The theorem follows from the fact that \minvc
has no  polynomial-time approximation scheme on 3-regular graphs, unless $\p=\np$ \cite{AK00}.
\qedfill
\end{proof}
}

In \autoref{th:whardness-general} we showed that \influence with unanimity thresholds is $\wone$-hard w.r.t. parameters $k$ and $\ell$.
In the following we give several fixed-parameter tractability results for \influence w.r.t. parameter~$k$ on regular graphs and bounded degree graphs with unanimity thresholds.
First we show that using results of Cai \emph{et al.}~\cite{cai2006b} 
we can obtain fixed-parameter tractability. Then we establish an explicit and more efficient combinatorial algorithm.
Using \cite{cai2006b} we can show:

\begin{theorem}\label{thm:UnanimityBoundedDegreeFPT}
	\sloppy \influence with unanimity thresholds can be solved in $2^{O(k \Delta^3)} n^2 \log n$ time where~$\Delta$ denotes the maximum degree
and  in~$2^{O(k^2 \log k)} n \log n$ time for regular graphs.
\end{theorem}

{
\begin{proof}
For graphs of maximum degree $\Delta$, we simply apply the result from \cite[Theorem 4]{cai2006b} with $i=3$.

Let $G$ be a $\Delta$-regular graph. When  $\Delta >k$, any~$k$ vertices of the graph form a  solution since no vertex outside the set becomes active.
	Hence, we assume in the following that~$\Delta \le k$.
Since $G$ is regular, it   follows that any subset $S$, $|S|=k$ can activate at most~$k$ vertices.
	Hence, the graph~$G[\sigma[S]]$ contains at most~$2k$ vertices and, thus, $\ell \le k$.
	Furthermore, since we consider unanimity thresholds, every vertex $v \in \sigma(S)$ has exactly~$\Delta$ neighbors in~$S$ and, thus, $|N_{G[\sigma[S]]}(v)| = \Delta$ and~$N_{G[\sigma[S]]}(v) \subseteq S$.
Our fpt-algorithm solving \influence runs  in two phases:
	\begin{description}
		\item[Phase 1:] Guess 
a graph~$H$ being isomorphic to~$G[\sigma[S]]$.
		\item[Phase 2:] Check whether~$H$ is a subgraph of~$G$.
	\end{description}
	Phase~1 is realized by simply iterating over all possible graphs~$H$ with $k+\ell$ vertices.
	A simple upper bound on the number of different graphs with~$k+\ell$ vertices is~$2^{\binom{k+\ell}{2}} \le 2^{4k^2}$.
	Hence, in Phase~1 the algorithm tries at most~$O(2^{4k^2})$ possibilities.
	Note that Phase~2 can be done in~$2^{O(\Delta k \log k)} n \log n$ using a result from~\cite[Theorem 1]{cai2006b}.
	Altogether this gives a running time \sloppy of~$O(2^{4k^2} 2^{O(\Delta k \log k)} n \log n)$.
	Since~$\Delta \le k$, this gives~$2^{O(k^2 \log k)} n \log n$. The correctness of the algorithm follows from the exhaustive search.
%
\qedfill\end{proof}
}

While the previous results use general frameworks to solve the problem, we now give a direct combinatorial algorithm for \influence with unanimity thresholds on bounded degree graphs. 
For this algorithm we need the following definition and lemma.

\begin{definition}
	Let~$(\alpha,\beta)$ be a pair of positive integers, $G=(V,E)$  an undirected graph with unanimity thresholds, and $v\in V$ a vertex.
	A vertex~$v$ is called a \emph{realizing} vertex for the pair~$(\alpha,\beta)$ if there exists a vertex subset~$V' \subseteq N^{2\alpha-1}[v]$ of size~$|V'| \le \alpha$ such that~$|\sigma(V')| \ge \beta$ and~$\sigma[V']$ is connected.
	Furthermore, $\sigma[V']$ is called a \emph{realization} of the pair~$(\alpha,\beta)$.
\end{definition}
We show first  that in bounded degree graphs the problem of deciding whether a vertex is a realizing vertex for a pair of positive integers~$(\alpha,\beta)$ is fixed-parameter tractable w.r.t. parameter~$\alpha$.

\begin{lemma} \label{lem:connectedCase}
	Checking whether a vertex~$v$ is a realizing vertex for a pair of positive integers~$(\alpha,\beta)$ can be done in~
$\Delta^{O(\alpha^2)}$ time, where~$\Delta$ is the maximum degree.
\end{lemma}
{
\begin{proof}
	The algorithm solving the problem checks for all vertex subsets~$V'$ of size~$\alpha$ in~$N^{2\alpha - 1}[v]$ whether~$V'$  activates at least~$\beta$ vertices and whether~$\sigma[V']$ is connected.
	Since we consider unanimity thresholds it follows that~$\sigma[V'] \subseteq N^{2\alpha}[v]$.

	The correctness of this algorithm results from the exhaustive search.
	We study in the following the running time:
	The~$(2\alpha - 1)^{\text{th}}$ neighborhood of any vertex contains at most~$\Delta(\Delta^{2\alpha}) / (\Delta - 1) + 1 \le 2\Delta^{2\alpha}$ vertices.
	Hence, there are~$2^\alpha\Delta^{(2\alpha)\alpha}$ possibilities to choose the~$\alpha$ vertices forming~$V'$.
	For each choice of~$V'$ the algorithm has to check how many vertices are activated by~$V'$.
	Since this can be done in linear time and there are~$O(\Delta \Delta^{2\alpha})$ edges, this gives another~$O(\Delta^{2\alpha+1})$ term.
	Altogether, we obtain   a running time of~$O(2^\alpha\Delta^{2\alpha^2 + 2\alpha + 1})$ = $\Delta^{O(\alpha^2)}$.
	\qedfill
\end{proof}
}
Consider in the following the \conninfluence problem that is \influence  with the additional requirement that~$G[\sigma[S]]$ has to be connected.
Note that with \autoref{lem:connectedCase} we can show that \conninfluence is fixed parameter tractable w.r.t. parameter~$k$ on bounded degree graphs.
Indeed,  observe that two vertices in~$\sigma(S)$ cannot be adjacent since we consider unanimity thresholds.
From this and the requirement that~$G[\sigma[S]]$ is connected, it follows that~$G[\sigma[S]]$ has a diameter of at most~$2k$.
Hence, the algorithm for \conninfluence checks for each vertex~$v \in V$ whether~$v$ is a realizing vertex for the pair~$(k,\ell)$.
By \autoref{lem:connectedCase} this gives an overall running time of~$\Delta^{O(k^2)} \cdot n$.

We can extend the algorithm for the connected case to deal with the case where~$G[\sigma[S]]$ is not connected.
The general idea is as follows.
For each connected component~$C_i$ of~$G[\sigma[S]]$ the algorithm guesses the number of vertices in~$S \cap C_i$ and in~$\sigma(S) \cap C_i$.
This gives an integer pair~$(k_i,\ell_i)$ for each connected component in~$G[\sigma[S]]$.
Similar to the connected case, the algorithm will determine realizations for these pairs and the union of these realizations give~$S$ and~$\sigma(S)$.
Unlike the connected case, it is not enough to look for just one realization of a pair~$(k_i,\ell_i)$ since the realizations of different pairs may be not disjoint and, thus, vertices may be counted twice as being activated.
To avoid the double-counting we show that if there are ``many'' different realizations for a pair~$(k_i,\ell_i)$, then there always exist a realization being disjoint to all realizations of the other pairs.
Now consider only the integer pairs that do not have ``many'' different realizations.
Since there are only ``few'' different realizations possible, the graph induced by all the vertices contained in all these realizations is ``small''.
Thus, the algorithm can guess the realizations of the pairs having only ``few'' realizations and afterwards add greedily disjoint realizations of pairs having ``many'' realizations.
See \autoref{alg:PseudocodeFPTinfluence} for the pseudocode.

\begin{algorithm}[!ht]
 \small{
 \begin{algorithmic}[1]
	\Procedure{\textsc{solveInfluence}}{$G$, $\thr$, $k$, $\ell$}
		\State Guess~$x \in \{1, \ldots, k\}$ \Comment{$x$: number of connected components of~$G[\sigma[S]]$} \label{line:guess_x}
		\State Guess~$(k_1, \ell_1), \ldots, (k_x, \ell_x)$ such that~$\sum_{i=1}^x k_i = k$ and~$\sum_{i=1}^x \ell_i = \ell$ \label{line:guessPairs} 
		\State Initialize~$c_1 = c_2 = \ldots = c_x \gets 0$ \Comment{one counter for each integer pair~$(k_i,\ell_i)$} \label{line:initCounters}
		\For{each vertex~$v \in V$} \Comment{determine realizing vertices} \label{line:allVerticesRealizationsCheck}
			\For{$i \gets 1$ \textbf{to}~$x$}
				\If{$v$ is a realizing vertex for the pair~$(k_i,\ell_i)$} \Comment{see \autoref{lem:connectedCase}} \label{line:checkPairRealization}
					\State $c_i \gets c_i + 1$
					\State $T(v,i) = $ ``yes''
				\Else\
					\State $T(v,i) = $ ``no''
				\EndIf
			\EndFor
		\EndFor
		\State initialize~$X \gets \emptyset$ \Comment{$X$ stores all pairs with ``few'' realizations}
		\For{$i \gets 1$ \textbf{to}~$x$}  \label{line:determinePairsWithFewRealization}
			\If{$c_i \le 2 \cdot x \cdot \Delta^{4k}$}
				\State $X \gets X \cup \{i\}$
			\EndIf
		\EndFor
		\For{each vertex~$v \in V$} \Comment{remove vertices not realizing any pair in $X$} \label{line:deleteVertices}
			\If{$\forall i\in X: T(v,i) = $ ``no''}
				\State delete~$v$ from~$G$.
			\EndIf
		\EndFor
		\If{all pairs~$(k_i,\ell_i)$, $i \in X$, can be realized in the remaining graph} \label{line:checkRemainingRealization}
			\State \Return `YES'
		\Else\
			\State \Return `NO'
		\EndIf
	\EndProcedure
 \end{algorithmic}
 }
 \caption{The pseudocode of the algorithm solving the decision problem \influence. The guessing part in the algorithm behind \autoref{lem:connectedCase} is used in Line~\ref{line:checkPairRealization} as subroutine. The final check in Line~\ref{line:checkRemainingRealization} is done by brute force checking all possibilities.}
 \label{alg:PseudocodeFPTinfluence}
\end{algorithm}


\begin{theorem}\label{thm:unanimity-bounded-degree-fpt}
	\autoref{alg:PseudocodeFPTinfluence} solves \influence with unanimity thresholds in~$2^{O(k^2 \log (k\Delta))} \cdot n$ time, where $\Delta$ is the maximum degree of the input graph.
\end{theorem}
{
\begin{proof}
	Let~$S$ be a solution set, that is, $S \subset V$, $|S| \le k$ and~$\sigma(S) \ge \ell$.
	In the following we show that \autoref{alg:PseudocodeFPTinfluence} decides whether or not such set~$S$ exists in~$2^{O(k^2 \log (k\Delta))} \cdot n$ time.
	We remark that the algorithm can be adapted to also give such set~$S$ if it exists. 
	First we prove the correctness of the algorithm and then show the running time bound.

	\emph{Correctness:} We now show that a solution set~$S$ exists if and only if the algorithm returns ``YES''.
	``$\Rightarrow:$''
	Assume that~$S$ is the solution set.
	Observe that~$G[\sigma[S]]$ consists of at most~$k$ connected components and, thus, the guesses in Lines~\ref{line:guess_x} and~\ref{line:guessPairs} are correct.
	Clearly, in the solution set~$S$ there is a realization for each pair~$(k_i, \ell_i)$.
	Furthermore observe that in Line~\ref{line:determinePairsWithFewRealization} it holds that~$X \subseteq \{1, \ldots, x\}$ and that in the loop starting in Line~\ref{line:deleteVertices} only vertices that cannot realize any pair corresponding to~$X$ are deleted.
	Hence, there exists a realization for the pairs corresponding to~$X$ in the remaining graph.
	Since the checking in Line~\ref{line:checkRemainingRealization} is done by trying all possibilities, the algorithm returns ``YES''.
	
	``$\Leftarrow:$'' Now assume that the algorithm returns ``YES''.
	Observe that this implies that in Line~\ref{line:checkRemainingRealization} there exists a realization for the all the pairs corresponding to~$X$.
	Hence, it remains to show that for each pair~$(k_j, \ell_j)$ where $j \in \{1, \ldots , x\} \setminus X$ there exists a realization in $G$.
	(Clearly, if all pairs are realized then the union of the realizations form the vertex set $\sigma[S]$ such that $|S| = k$.)
	To see that there exist realizations for these pairs observe the following:
	The $(4k)^{\text{th}}$ neighborhood of any vertex contains at most $2\Delta^{4k}$ vertices.
	Thus, if in the case of two pairs $(k_1 , \ell_1), (k_2, \ell_2)$ the value of the second counter is $c_2 > 2\Delta^{4k}$, then we can deduce that for every realizing vertex~$v_1$ for~$(k_1, \ell_1)$ there exists a realizing vertex $v_2$ for $(k_2 , \ell_2)$ such that the distance~$d$ between~$v_1$ and~$v_2$ is more than~$4k$.
	Since $d > 4k$, it follows that the realizations for $(k_1 , \ell_1)$ and~$(k_2 , \ell_2)$ do not overlap.
	(If two realizations would overlap then some vertices in $\sigma(S)$ may be counted twice.)
	Generalizing this argument to $x$ integer pairs $(k_1 , \ell_1 ), \ldots , (k_x , \ell_x)$ yields the following: If there exists an $i \in \{1, \ldots , x\}$ such that $c_i > x \cdot 2 \cdot \Delta^{4k}$, then for any realization of the pairs $(k_j , \ell_j )$ with $i \neq j$ there exists a non-overlapping realization of $(k_i , \ell_i )$.
	Thus, we can ignore the pair $(k_i, \ell_i)$ where $c_i > x \cdot 2 \cdot \Delta^{4k}$ in the remaining algorithm and can assume that $(k_i , \ell_i )$ is realized.
	
	Observe that from the Lines~\ref{line:allVerticesRealizationsCheck} to~\ref{line:deleteVertices} it follows that for all $j \in \{1, \ldots , x\} \setminus X$ we have $c_j > x \cdot 2 \cdot \Delta^{4k}$.
	Thus, from the argumentation in the previous paragraph it follows that there exist non-overlapping realizations for all pairs corresponding to $\{1, \ldots , x\} \setminus X$.
	Thus, there exists a solution set~$S$ as required.
	
	\emph{Running time:}
	Observe that~$\ell \le \Delta k$ as described in the proof of \autoref{lem:boundedDegAffect}.
	Thus, the guessing in Lines~\ref{line:guess_x} and~\ref{line:guessPairs} can clearly be done in~$O(k \cdot k^k (\Delta k)^k) = O(k^{2k+1} \Delta^{k})$.
	By \autoref{lem:connectedCase} the checking in Line~\ref{line:checkPairRealization} can be done in~$\Delta^{O(k_i^2)}$ time.
	Thus, the loop in Line~\ref{line:allVerticesRealizationsCheck} requires~$n \cdot \sum_{i=1}^x \Delta^{O(k_i^2)} \le \Delta^{O(k^2)} \cdot x \cdot n$ time.
	Clearly, the loop in Line~\ref{line:determinePairsWithFewRealization} needs~$O(x) \le O(k)$ time.
	Furthermore, the loop in Line~\ref{line:deleteVertices} needs~$O(k \cdot n)$ time.
	For the checking in Line~\ref{line:checkRemainingRealization} observe the following.
	After deleting the vertices in the loop in Line~\ref{line:deleteVertices} the remaining graph can have at most~$\sum_{i \in X}c_i \le x \cdot 2\cdot x \cdot \Delta^{4k}$ vertices.
	Furthermore, $\sum_{i \in X}k_i \le k$ and, thus, there are at most~$(2\cdot x^2 \cdot \Delta^{4k})^k$ candidate subsets for the solution set~$S$.
	Checking whether~$\sum_{i \in X}k_i$ chosen vertices activate~$\sum_{i \in X}\ell_i$ other vertices can be done in~$(2\cdot x^2 \cdot \Delta^{4k})^2$ time.
	Hence, the checking in Line~\ref{line:checkRemainingRealization} can be done in~$\Delta^{O(k^2)}$ time.
	Putting all together we arrive at a running time of~$(k\Delta)^{O(k^2)} \cdot n = 2^{O(k^2 \log(k\Delta))} \cdot n$.
	\qedfill
\end{proof}
}

\section{Conclusions} \label{sec:conclusions}
We established results concerning the parameterized complexity as well as the polynomial-time and fpt-time approximability of two problems modeling the spread of influence in social networks, namely \maxopinfluence and \maxclinfluence.

In the case of unanimity thresholds, we show that \maxopinfluence  is at least as hard to approximate as \densestksubgraph, a well-studied problem. We established that \densestksubgraph is $r(n)$-approximable for any  strictly increasing function $r$ in fpt-time w.r.t. parameter $k$.  An interesting open question consists of determining whether \maxopinfluence  is constant approximable in fpt-time. Such a positive result would improve the approximation in fpt-time for \densestksubgraph.
%
In the case of thresholds bounded by two we excluded a polynomial time approximation scheme for \maxclinfluence but we did not found any polynomial-time approximation algorithm.
Hence, the question arises, whether this hardness result can be strengthened.
%





\bibliographystyle{plain}
\bibliography{main}



\end{document}